\pdfoutput=1 
\def\year{2022}\relax
\documentclass[letterpaper]{article} 
\usepackage{aaai22}  
\usepackage{times}  
\usepackage{helvet}  
\usepackage{courier}  
\usepackage[hyphens]{url}  
\usepackage{graphicx} 
\usepackage{natbib}  
\usepackage{caption} 
\DeclareCaptionStyle{ruled}{labelfont=normalfont,labelsep=colon,strut=off} 
\usepackage{amsmath,amsfonts,amsthm}
\usepackage{mathtools}
\newtheorem{theorem}{Theorem}
\newtheorem{define}{Definition}
\newtheorem{lemma}{Lemma}
\newtheorem{assume}{Assumption}
\newcommand{\dif}{\mathrm{d}}
\DeclareMathOperator*{\argmax}{argmax}

\title{Characterization of Incentive Compatibility of an Ex-Ante Constrained Player}
\author {
    Bonan Ni,
    Pingzhong Tang 
}
\affiliations {
    Tsinghua University\\
    bricksern@gmail.com, kenshinping@gmail.com
}

\begin{document}

\maketitle

\begin{abstract}

We consider a variant of the standard Bayesian mechanism, where players evaluate their outcomes and constraints in an {\em ex-ante} manner. Such a model captures a major form of modern online advertising where an advertiser is concerned with her/his expected utility over a time period and her/his type may change over time. We are interested in the incentive compatibility (IC) problem of such Bayesian mechanism. Under very mild conditions on the mechanism environments, we give a full characterization of IC via the {\em taxation principle} and show, perhaps surprisingly, that such IC mechanisms are fully characterized by the so-called {\em auto-bidding} mechanisms, which are pervasively fielded in the online advertising industry.


\end{abstract}

\section{Introduction}

In online advertisement industry, where a publisher sells advertisement slots to some advertisers, it is typical for each advertiser to have her own account which records her remaining budget, the amount of advertisement slots allocated so far and possibly the amount of conversions generated through her advertisements. Slot auctions happen at almost the same frequency as user requests, which makes it almost impossible (and indeed unnecessary) for advertisers to evaluate utility of each auction separately. As a result, in industry, certain kinds of contracts are signed between the platform and the advertiser to guarantee overall allocation and payment over a relatively long time period, typically one day or one week. For example, one popular form of contract is to limit the total amount of money spent on advertising, which is called budget constraint in the related literature. Another popular contract form limits the return-over-investment ratio (ROI), and examples of this form include Google's target CPA bidding and Facebook's cost cap. 

In the probability space consisting of all auction events over the given time period, the constraints imposed by such contracts are expressed in an ex-ante manner. That is, satisfaction of the constraint is evaluated by taking the expectation over the whole probability space. Most constraints in classic mechanism design literature, however, are {\em interim} constraints, which are separately evaluated for each possible type of each player. When an ex-ante constraint on overall expectation is binding, it is possible for the player to design more complicated misreport strategy: instead of separately maximizing the utility for each of her own type, the player can make contingent plans about the type report over her whole type space, so that her constraint remains satisfied, but her utility is overall better off than always reporting truthfully. Such possibility poses new challenges to the design of incentive compatible mechanisms.

In this paper, we analyze the general mechanism design problem with an ex-ante constraint on the outcome and payment, which is not limited to the design of single-item auctions. From a fixed player's perspective, the mechanism is abstracted into a pair of interim rules (but with ex-ante utility model and constraints), and we characterize interim rules of any incentive compatible mechanism. Our first characterization in Section \ref{char-1} states that, it is without loss of generality to use the following implementation of incentive compatible mechanism: the designer maximizes some new value-minus-payment utility on behalf of the player for every type realization of the mechanism, where the new value is defined by linearly combining a utility term and a constraint term using some factors that are fixed across all type realization of the mechanism. For example, when the player is a value maximizer with budget constraint, it suffices to determine factors $(\alpha, \beta)$ and maximize $\alpha$ times value minus $\beta$ times payment. This characterization justifies the auto-bidding technique in online advertising, where an advertiser's evaluations for the advertisement slot's value and contribution to constraint are combined linearly to give input to an auction. Our second characterization, given in Section \ref{char-2}, shows that incentive compatibility can be characterized using a multi-dimensional payment identity, and the new payment identity differs from the one in classic quasi-linear mechanism design literature \cite{10.2307/3689266,MANELLI2007153} only by linearly combining the utility and constraint terms.

\section{Related works}\label{relate-work}

Many recent works study the the mechanism design problems for players with various ex-ante constraints, typically under the context of online advertisement auction. Unlike the widely used constraint of individual rationality which is separately evaluated for every possible type of the player (interim IR) \cite{10.1145/2940716.2940775} or every profile of all players' types (ex-post IR) \cite{Ashlagi2019SequentialMW}, an ex-ante constraint is evaluated by taking expectation over the whole game's probability space. Two of the most widely discussed ex-ante constraints are budget constraint \cite{10.1145/3465456.3467578,10.1145/1367497.1367747}, which is to limit the overall expectation of payment (the term `budget constraint' is also used for interim or ex-post constraint on payment \cite{principal-agent}, and the budget level can be private \cite{10.1111/1467-937X.00033}, but in this paper we focus on an ex-ante constraint known to the designer), and ROI constraint, which is to limit the ratio given by the player's expected utility over her expected payment \cite{10.1145/3465456.3467607,10.1145/3442381.3449841,10.1145/3338506.3340269}. Given the constraint, the player seeks to maximize her utility without violating the constraint. In many works the player's utility is still given by the standard quasi-linear utility, which is the value she 
wins from the outcome minus payment. Another player model which becomes increasingly popular assumes that the player's utility is just the value she wins, and the trade-off between value and payment is completely controlled by the ex-ante constraint \cite{Babaioff2021NonquasilinearAI}. This model plus a budget constraint resembles the classic model of Fisher's market \cite{https://doi.org/10.1111/j.1536-7150.2005.00349.x}, where the player maximizes the value of items she wins, under the constraint that the money spent on the items does not exceed a given limit.

To solve the mechanism design problem with ex-ante constrained players, the technique called auto-bidding is widely proposed \cite{10.1145/3465456.3467607,10.1007/978-3-030-35389-6_2,10.1145/3442381.3450052,DBLP:journals/corr/abs-2106-06224,li2020incentive}. Typically, the mechanism is implemented by running an auction which repeatedly asks for a single bid from the player, and decide allocation and payment according to the bid. Auto-bidding repeatedly decides the bid on behalf of the player, so that the total utility over multiple rounds is maximized without violating the constraint. An important result is that, under certain settings, the optimal bidding rule is to always linearly combine the current item's contribution to utility and contribution to the constraint using some factors that are fixed across all rounds. Many previous works obtain this result under different restrictions: \cite{10.1007/978-3-030-35389-6_2} assumes finite rounds of game and a posted price auction, and \cite{10.1145/3465456.3467607,Babaioff2021NonquasilinearAI} assumes specific forms of the player's utility and the ex-ante constraint. Some other works capture the sequential revelation of players' types, and design dynamic mechanisms or control strategies that achieve low regret over multiple rounds  \cite{golrezaei2021bidding,regret2019}.

Another literature investigates the technique called pacing. Pacing can be similarly seen as certain control technique deciding input to a repeated auction, but its goal is often to smooth the usage of the player's budget over a certain time period \cite{Agarwal2014BudgetPF}. Smoothing the budget usage does not necessarily involve the maximization of the player's utility, which makes it  less natural to define the corresponding incentive compatibility problem. However, some works define pacing as trying to use up the budget by uniformly scaling the bid \cite{Conitzer2018MultiplicativePE,inproceedings}, which  actually maximizes the player's utility given some certain player models.

\section{Model}\label{mechanism-define}
To study the incentive compatibility of the mechanism, we focus our attention to an arbitrarily chosen player and analyze the incentive problems from her own perspective, without concerning any component of the mechanism that does not affect the chosen player's action. Methodologically, all analysis in the paper considers the mechanism's interim rules with respect to a single fixed player. When there are multiple players, all the results in this paper can be applied to each of them, so that each player has her corresponding incentive compatibility characterization about her interim rules. 

Denote by $V \subseteq \mathbb{R}^d$ the player's type space, by $Q \subseteq \mathbb{R}^D$ the mechanism's outcome space. When analyzing incentive compatibility, it is without loss of generality to consider only direct mechanisms, which take the player's type report as the only input from the player. Therefore, from the player's perspective, we consider a pair of interim rules $(x, p)$ as the mechanism's abstraction, where $x: V \rightarrow Q$ denotes the mapping from type report to outcome, and $p: V \rightarrow \mathbb{R}$ denotes the mapping from type report to expected payment. The outcome space $Q$ are chosen so that both $x$ and $p$ are deterministic functions: any possible lottery over outcomes must be encoded by a corresponding single point in $Q$. For example, if we choose $Q = [0,1]^2$, any possible lottery to allocate two items to the player can be encoded by a single point in $Q$.

Denote by $\Delta V$ the set of Borel probability measures on $V$. The player's type distribution is given by some probability measure $\rho \in \Delta V$. We say that a subset $E \subseteq V$ has positive measure if $\rho(E)>0$, and a condition holds for almost every $v$ if every subset where the condition fails does not have positive measure. Unless otherwise specified, all expectations in this paper are taken with respect to $\rho$.

Denote by $s: V \rightarrow \Delta V$ a strategy the player can take. The player model is given by the combination of a utility function and an ex-ante constraint, and is encoded by tuple $(u, f, c_1, c_2, C)$. We assume that the player is risk-neutral, so that the utility is also evaluated by taking ex-ante expectation. The utility extracted from the outcome is given by function $u: Q \times V \rightarrow \mathbb{R}$, and overall utility is given by the ex-ante expectation of utility from outcome minus a constant $c_1 \in \{0,1\}$ times payment. Given interim rules $(x,p)$ and report strategy $s$, the ex-ante utility is given by 
\begin{equation} \label{utility}
    U(x, p, s) \coloneqq \mathbb{E}_{v, v' \sim s(v)} [u(x(v'), v) - c_1 \cdot p(v')].
\end{equation}
Note that for any $c_1 > 0$, we can rescale $u$ into $u' = u/c_1$, so that maximizing $\mathbb{E}[u(x(v'), v) - c_1 \cdot p(v')]$ is equivalent to maximizing $\mathbb{E}[u'(x(v'), v) - p(v')]$. Therefore, it is without loss of generality to consider $c_1 \in \{0,1\}$ instead of $c_1 \ge 0$.

The {\em ex-ante constraint} is decided by function $f: Q \times V \rightarrow \mathbb{R}$, constant $c_2 \in \{0, 1\}$ and constant $C \in \mathrm{R}$. The constraint is given by:
\begin{equation} \label{constraint}
    \mathbb{E}_{v, v' \sim s(v)} [f(x(v'), v) - c_2 \cdot p(v')] \ge C.
\end{equation}
And we say that the ex-ante constraint is binding when reporting truthfully, if 
\[
    \mathbb{E}_{v} [f(x(v), v) - c_2 \cdot p(v)] = C
\]

Throughout this paper, we only analyze strategy $s$ for which the expectation in equation (\ref{utility}) and (\ref{constraint}) exist. In fact, it is sufficient to consider strategies which are some random lotteries over a finite number of deterministic report rules $h: V \rightarrow V$, with each $h$ being measurable. 

An ex-ante constraint can be thought of as a contract between the mechanism designer and the player, so that even without any prior distributional knowledge, the player is convinced that the overall outcome and payment are under some specific control. Our definition of incentive compatibility requires that the constraint is always satisfied by truth-telling, therefore an alternative interpretation is that when violating the constraint, the player's utility is $-\infty$. This player model covers a wide range of buyer models in online advertising and auto-bidding literature, including:
\begin{itemize}
    \item Value maximizer with ROI constraint: letting $c_1=0, f=u / (1+\gamma), c_2 = 1, C=0$, where $\gamma$ is the lowest acceptable ROI for the buyer.
    \item Utility maximizer with budget constraint: letting $c=1, f=0$, $c_2 = 1$, and $-C$ be the buyer's budget limit.
\end{itemize}
Also, the model naturally covers the buyer of a multi-item auction. With $D$ equal to the number of items and $d=2^D$, we can use each dimension of $V$ to represent the player's valuation for a subset of the items, and each dimension of $Q$ to represent the allocation probability of an item.

We make the following technical assumptions on $V$, $Q$, functions of the player model $u, f$, and the interim rules $(x,p)$.
\begin{assume} \label{continuity}
    $V$ is a compact subset of $\mathbb{R}^d$, and $Q$ is a compact subset of $\mathbb{R}^D$. $u, f, x, p$ are all continuous.
\end{assume}
By assuming compact $V,Q$ and continuous $u, f, x, p$, we make the values of $u, f, p$ bounded. 

We focus on deriving characterizations of the mechanism's incentive compatibility. The incentive compatibility for the player with ex-ante constraint is defined as follows:
\begin{define}
    An interim rule pair $(x,p)$ is {\em incentive compatible}, if 
    \[
        \mathbb{E}_{v} [f(x(v), v) - c_2 \cdot p(v)] \ge C,
    \]
    AND, for any strategy $s: V \rightarrow \Delta V$ s.t. the expectations $U(x,p,s)$ and $\mathbb{E}_{v, v' \sim s(v)} [f(x(v'), v) - c_2 \cdot p(v')]$ exist, at least one of the following two inequalities hold:
    \begin{enumerate}
        \item $\mathbb{E}_{v} [u(x(v), v) - c_1 \cdot p(v)] \ge U(x, p, s)$.
        \item $\mathbb{E}_{v, v' \sim s(v)} [f(x(v'), v) - c_2 \cdot p(v')] < C$.
    \end{enumerate}
\end{define}
As explained in the beginning of the section, we abstract the mechanism into a pair of interim rules $(x,p)$. For incentive compatible interim rule pair $(x,p)$, reporting truthfully does not violate the ex-ante constraint, and maximizes utility among all report strategies that does not violate the constraint.

\section{Auto-bidding and taxation principle}\label{char-1}
Following the idea of taxation principle, we give our first characterization of incentive compatibility. Denote by $Q_x \subseteq Q$ the range of $x$. The classic taxation principle tells that, incentive compatibility mechanisms can be implemented by designing payment rule that is some mapping from $Q_x$ to $\mathbb{R}$, and decide outcome by choosing the optimal outcome in $Q_x$ on behalf of the player. Theorem \ref{taxation} tells that, a similar implementation can be used for player with an ex-ante constraint:
\begin{theorem} \label{taxation}
    Under Assumption \ref{continuity}, interim rule pair $(x,p)$ is incentive compatible if and only if at least one of the following two conditions holds:
    \begin{enumerate}
        \item The ex-ante constraint is binding when reporting truthfully, and there exists function $p_x: Q_x \rightarrow \mathbb{R}$ s.t. for almost every $v \in V$, 
        \[
            p_x(x(v)) = c_2 \cdot p(v),
        \]
        and for almost every $v \in V$,
        \begin{equation}
            x(v) \in \argmax_{q \in Q_x} [f(q, v) - p_x(q)],
        \end{equation}
        and if there exists $v' \ne v$ s.t. $f(x(v), v) - p_x(x(v)) = f(x(v'), v) - p_x(x(v'))$, then we have $u(x(v), v) - c_1 \cdot p(v) \ge u(x(v'), v) - c_1 \cdot p(v')$.
        \item There exists constant $r \in \mathbb{R}_{\ge 0}$ and function $p_x: Q_x \rightarrow \mathbb{R}$ s.t. for almost every $v \in V$, 
        \[
            p_x(x(v)) = (c_1+r c_2)p(v),
        \]
        and for almost every $v \in V$, 
        \begin{equation}\label{taxed-maximization}
            x(v) \in \argmax_{q \in Q_x} [u(q, v) + r \cdot f(q, v) - p_x(q)],
        \end{equation}
        and if $r>0$, the ex-ante constraint is binding when reporting truthfully.
    \end{enumerate}
\end{theorem}
Theorem \ref{taxation} can be seen as a generalization of the widely used auto-bidding technique: the game runs by offering a price $p_x(q)$ for every possible outcome $q \in Q_x$, and to report optimally on behalf of the player, the controller either maximizes contribution to constraint per round $[f(\cdot, v) - p_x(\cdot)]$, or maximizes a goal which is some fixed linear combination of the utility and the contribution to constraint. We show in Theorem \ref{ic-iff} that, under some additional assumptions, the corresponding interim rules are characterized by an equality similar to the multi-dimensional version of standard payment identity, which generalizes the discovery in auto-bidding literature, that the underlying auction mechanism should be truthful for quasi-linear players.

To prove Theorem \ref{taxation}, we first provide some useful lemmas. For any $r \ge 0$, define set $V^+(r)$ and $V^-(r)$ by
{\small
\begin{align*}
    V^+(r) \coloneqq & \big \{v |  v \in V, \exists v' \in V:\\
    & f(x(v'), v) - c_2 \cdot p(v') < f(x(v), v) - c_2 \cdot p(v), \\
    & [u(x(v'), v) - c_1 \cdot p(v')] - [u(x(v), v) - c_1 \cdot p(v)] \\
    & > -r \cdot [f(x(v'), v) - c_2 \cdot p(v')] \\
    & + r \cdot [f(x(v), v) - c_2 \cdot p(v)]\big \} \\
     V^-(r) \coloneqq & \big \{v | v \in V, \exists v' \in V:\\
    & f(x(v'), v) - c_2 \cdot p(v') > f(x(v), v) - c_2 \cdot p(v), \\
    & [u(x(v'), v) - c_1 \cdot p(v')] - [u(x(v), v) - c_1 \cdot p(v)] \\
    & >  -r \cdot [f(x(v'), v) - c_2 \cdot p(v')] \\
    & + r \cdot [f(x(v), v) - c_2 \cdot p(v)]\big \}
\end{align*}
}%

The following lemma tells that, such subsets of $V$ are measurable.
\begin{lemma} \label{measurable}
    Under Assumption \ref{continuity}, for any $r \ge 0$, $V^+(r)$ and $V^-(r)$ are both measurable.
\end{lemma}
\begin{proof}
    For any $r \ge 0$, we prove that $V^+(r)$ is open subset of $V$, and similarly $V^-(r)$ is open subset of $V$, therefore both sets are measurable. 

    By continuity of $x, p, f, u$, for any $v, v'$ s.t. $f(x(v'), v) - c_2 \cdot p(v') < f(x(v), v) - c_2 \cdot p(v)$ and $[u(x(v'), v) - c_1 \cdot p(v')] - [u(x(v), v) - c_1 \cdot p(v)] > -r \cdot [f(x(v'), v) - c_2 \cdot p(v')] + r \cdot [f(x(v), v) - c_2 \cdot p(v)]$, there exists some $\epsilon > 0$ so that for any $v'' \in V \cap B(v, \epsilon)$ where $B(v, \epsilon)$ is the open ball centered at $v$ with radius $\epsilon$, we have $f(x(v'), v'') - c_2 \cdot p(v') < f(x(v''), v'') - c_2 \cdot p(v'')$ and $[u(x(v'), v'') - c_1 \cdot p(v')] - [u(x(v''), v'') - c_1 \cdot p(v'')] > -r \cdot [f(x(v'), v'') - c_2 \cdot p(v')] + r \cdot [f(x(v''), v'') - c_2 \cdot p(v'')]$. This gives $v'' \in V^+(r)$, therefore $V^+(r)$ is open subset of $V$.
\end{proof}
And the following lemma allows us to construct some measurable deterministic deviation:
\begin{lemma}\label{exist-integral}
    Under Assumption \ref{continuity}, for any $E \subseteq V$ with $\rho(E)> 0$, for any $n > 0$, any $(\alpha_1, \alpha_2, \ldots,\alpha_n) \in \mathbb{R}^n$ and $(\beta_1, \beta_2, \ldots, \beta_n) \in \mathbb{R}^n$, suppose for every $v \in E$, there exists $v' \in V$ s.t. for any $1 \le i \le n$, we have that 
    \begin{align*}
        \gamma_i(v, v') \coloneqq & \alpha_i \cdot [u(x(v'), v) - c_1 \cdot p(v')] \\
        - & \alpha_i \cdot [u(x(v), v) - c_1 \cdot p(v)] \\
        + & \beta_i \cdot [f(x(v'), v) - c_2 \cdot p(v')] \\
        - & \beta_i \cdot [f(x(v), v) - c_2 \cdot p(v)] 
    \end{align*}
    is positive, then there exists function $h: E \rightarrow V$ s.t. for any $1 \le i \le n$ and function $g_i(v) \coloneqq \gamma_i(v, h(v))$, the integration $\int_{E} g_i d\rho$ exists and is positive.
\end{lemma}
\begin{proof}[Sketched proof]
    $\rho(E) > 0$ gives that, we can find some $\epsilon>0$ and some $E_\epsilon$ with $\rho(\epsilon) > 0$, so that for any $v \in E$, there exists $v' \in V$ s.t. any $1 \le i \le n$, we have $\gamma_i(v, v') \ge \epsilon$. By compactness of $V$, we can obtain that $V \setminus E_\epsilon$ is open in $V$, therefore $E_\epsilon$ is closed in $V$.
    
    For any $v \in E_\epsilon$ and any $1 \le i \le n$, since $x, p, u, f$ are all continuous, for every $v \in E_\epsilon$ we can find such an open set $E_v \subseteq V$ that contains $v$: there exists some $v'$ s.t. for every $1 \le i \le n$ and every $v'' \in E_v$, we have $\gamma_i(v'', v') >0$. The set
    \[
        \{E_v | v \in E_\epsilon\}
    \]
    is an open cover of $E_\epsilon$. Since $V$ is compact and $E_\epsilon$ is closed in $V$, we can find some finite subcover of $E_\epsilon$. By assigning each $v \in E_\epsilon$ to some element of the finite subcover, we can define function $h: E_\epsilon \rightarrow V$ whose range is finite. For $v \in E \setminus E_\epsilon$ we can simply make $h(v) = v$. This makes the integration of $g_i$ exist for every $1 \le i \le n$. 
\end{proof}

Given the above two lemmas, we are able to prove the following lemma, which is essential for our characterizations of incentive compatibility. In fact, the lemma itself gives if and only if condition of incentive compatibility.
\begin{lemma}\label{ic-iff-new}
    Under Assumption \ref{continuity}, the interim rule pair $(x,p)$ is incentive compatible, if and only if $\mathbb{E}_{v} [f(x(v), v) - c_2 \cdot p(v)] \ge C$, and at least one of the following two conditions holds:
    \begin{enumerate}
        \item The ex-ante constraint is binding when reporting truthfully, and for almost every $v \in V$, for every $v' \in V$, $f(x(v), v) - c_2 \cdot p(v) \ge f(x(v'), v) - c_2 \cdot p(v')$, and if $f(x(v), v) - c_2 \cdot p(v) = f(x(v'), v) - c_2 \cdot p(v')$, we have $u(x(v), v) - c_1 \cdot p(v) \ge u(x(v'), v) - c_1 \cdot p(v')$.
        \item There exists $r \ge 0$ s.t. for almost every $v \in V$, for every $v' \in V$, 
        {\small
        \begin{equation} \label{new-sup-g}
        \begin{aligned}
            & [u(x(v'), v) - c_1 \cdot p(v')] - [u(x(v), v) - c_1 \cdot p(v)] \\
            \le & -r \cdot [f(x(v'), v) - c_2 \cdot p(v')] \\
            & + r \cdot [f(x(v), v) - c_2 \cdot p(v)],
        \end{aligned}
        \end{equation}
        }%
        and if $r>0$, the ex-ante constraint is binding when reporting truthfully.
    \end{enumerate}
\end{lemma}
\begin{proof}
    \textbf{Sufficiency}: Suppose the first condition holds. A strategy $s$ obtains strictly higher utility than reporting truthfully only if there exists some $E \subseteq V$ with $\rho(E)>0$ s.t. for every $v \in E$, it reports some $v' \ne v$ with positive probability, and $\mathbb{E}_{v} [u(x(v), v) - c_1 \cdot p(v)] < \mathbb{E}_{v} [u(x(v'), v) - c_1 \cdot p(v')]$. However, for almost every $v \in E$ and any such $v'$ we have $\mathbb{E}_{v} [f(x(v), v) - c_2 \cdot p(v)] < \mathbb{E}_{v} [f(x(v'), v) - c_2 \cdot p(v')]$, which means such a strategy $s$ violates the ex-ante constraint. Therefore, no strategy $s$ gets higher utility than truthfully reporting without violating the constraint.

    Suppose the second condition holds. If $r=0$, reporting truthfully maximizes ex-post utility $[u(x(\cdot), v) - c_1 \cdot p(\cdot)]$ for any $v \in V$. Therefore, every $s$ does not obtain higher expected utility than truthfully reporting. 

    If $r > 0$, suppose there exists strategy $s$ s.t. reporting according to $s$ does not violate the ex-ante constraint, but gives strictly higher utility than reporting truthfully. This means, for any $r \ge 0$, we have
    \begin{align*}
        & \mathbb{E}_{v, v' \sim s(v)}[u(x(v'), v) - c_1 \cdot p(v')] \\
        >& \mathbb{E}_{v}[u(x(v), v) - c_1 \cdot p(v)], \\
        & r \cdot \mathbb{E}_{v, v' \sim s(v)}[f(x(v'), v) - c_2 \cdot p(v')] \\
        \le & r \cdot \mathbb{E}_{v}[f(x(v), v) - c_2 \cdot p(v)]
    \end{align*}
    By subtraction between the equations and rearrangement of terms, we have:
    \begin{align*}
        & \mathbb{E}_{v, v' \sim s(v)} [u(x(v'), v) - c_1 \cdot p(v')] \\
        & - \mathbb{E}_{v}[u(x(v), v) - c_1 \cdot p(v)] \\
        > & - r \cdot \mathbb{E}_{v, v' \sim s(v)}[f(x(v'), v) - c_2 \cdot p(v')] \\
        & + r \cdot \mathbb{E}_{v}[f(x(v), v) - c_2 \cdot p(v)]
    \end{align*}
    However, we have equation (\ref{new-sup-g}) for almost every $v \in V$ and every $v'$, which leads to contradiction. We conclude that such $s$ does not exist, which means $(x, p)$ is incentive compatible.

    \textbf{Necessity}: Suppose $(x, p)$ is incentive compatible, but both conditions in Lemma \ref{ic-iff-new} are violated. 

    If the ex-ante constraint is binding, violation of the first condition can be further divided into 2 cases:

    Case 1: There exists some $E \subseteq V$ with $\rho(E) > 0$ s.t. for every $v \in E$, there exists some $v' \in V$ s.t. $f(x(v), v) - c_2 \cdot p(v) = f(x(v'), v) - c_2 \cdot p(v')$, and $u(x(v), v) - c_1 \cdot p(v) < u(x(v'), v) - c_1 \cdot p(v')$. Then, by Lemma \ref{exist-integral} there exists some strategy which misreports on $E$ and reports truthfully on $V \setminus E$, which gives higher utility than reporting truthfully without violating the ex-ante constraint, and this contradicts with incentive compatibility.

    Case 2: There exists some $E \subseteq V$ with $\rho(E) > 0$ s.t. for every $v \in E$, there exists some $v' \in V$ with $f(x(v), v) - c_2 \cdot p(v) < f(x(v'), v) - c_2 \cdot p(v')$. Note that for almost every $v \in E$ and the corresponding $v'$, we have $u(x(v), v) - c_1 \cdot p(v) \ge u(x(v'), v) - c_1 \cdot p(v')$, otherwise there exists some strategy which misreports on $F$, so that the utility is higher than truthfully reporting, and the ex-ante constraint remains satisfied. 

    Since equation (\ref{new-sup-g}) does not hold for almost every $v \in V$ when $r=0$, there exists some $F \subseteq V$ with $\rho(F)>0$ s.t. for every $v \in F$, there exists some $v' \in V$ with $u(x(v), v) - c_1 \cdot p(v) < u(x(v'), v) - c_1 \cdot p(v')$. Similarly, for almost every $v \in F$ and the corresponding $v'$, we have $f(x(v), v) - c_2 \cdot p(v) > f(x(v'), v) - c_2 \cdot p(v')$.

    For any $r \ge 0$, define $\rho^+(r)$ and $\rho^-(r)$ by 
    \begin{align*}
        & \rho^+(r) \coloneqq \rho(V^+(r)),\\
        & \rho^-(r) \coloneqq \rho(V^-(r)).
    \end{align*}

    Note that $\rho^-$ is non-decreasing in $r$, while $\rho^+$ is non-increasing in $r$. Moreover, the existence of $F$ gives $\rho^+(0)>0$, and the existence of $E$ gives $\rho^-(r)>0$ for some sufficiently large $r$. If for any $r \ge 0$ we have $\rho^+ >0$, then there exists $r^*$ s.t. $\rho^+(r^*)>0, \rho^-(r^*)>0$. 

    If there exists $r$ s.t. $\rho^+(r) = 0$, define 
    \[
        r_0 \coloneqq \inf\{r | \rho^+(r) =0\}.
    \]
    Note that we have $\rho^+(r_0) =0$. Suppose we have $\rho^+(r_0)>0$, then for sequence $(r_i)_{i \ge 1}$ given by $r_i = r_0 - 2^{-i}$, we have $V^+(r_1) \subseteq V^+(r_2) \subseteq \cdots $ and $\cup_{i \ge 1}V^+(r_i) = V^+(r_0)$, therefore $\lim_{i \to \infty}\rho^+(r_i) = \rho^+(r_0)>0$, and there exists $i>0$ s.t. $\rho^+(r_0-2^{-i})>0$, which contradicts with the definition of $r_0$. Suppose $\rho^-(r_0) =0$, then equation (\ref{new-sup-g}) holds for almost every $v \in V$ given $r = r_0$, thus the second condition is satisfied. Otherwise, $\rho^-(r_0) >0$ implies the existence of some $r^*<r_0$ s.t. $\rho^+(r^*)>0, \rho^-(r^*) >0$. 

    Therefore, if both conditions are not satisfied, we can get $r^*$ with $\rho^+(r^*)>0, \rho^-(r^*)>0$. Define $F_1 \coloneqq V^+(r^*)$ and $F_2 \coloneqq V^-(r^*)$, both with positive measure, and there exists functions $h_1: F_1 \rightarrow V$, $h_2: F_2 \rightarrow V$ s.t. by Lemma \ref{exist-integral}, for changes in utilities and LHS of ex-ante constraint defined as
    \begin{align*}
        \Delta U_1 = & \mathbb{E}_{v \in F_1} [u(x(h_1(v)),v) - c_1 \cdot p(h_1(v))] \\
        & - \mathbb{E}_{v \in F_1} [u(x(v),v) - c_1 \cdot p(v)], \\
        \Delta U_2 = & \mathbb{E}_{v \in F_2} [u(x(h_2(v)),v) - c_1 \cdot p(h_2(v))] \\
        & - \mathbb{E}_{v \in F_2} [u(x(v),v) - c_1 \cdot p(v)], \\
        \Delta C_1 = & \mathbb{E}_{v \in F_1} [f(x(h_1(v)),v) - c_2 \cdot p(h_1(v))] \\
        & - \mathbb{E}_{v \in F_1} [f(x(v),v) - c_2 \cdot p(v)], \\
        \Delta C_2 = & \mathbb{E}_{v \in F_2} [f(x(h_2(v)),v) - c_2 \cdot p(h_2(v))] \\
        & - \mathbb{E}_{v \in F_2} [f(x(v),v) - c_2 \cdot p(v)],
    \end{align*}
    we have $\Delta U_1 > 0, \Delta C_2 > 0$, and $\Delta U_1 + r^* \cdot \Delta C_1 > 0, \Delta U_2 + r^* \cdot \Delta C_2 > 0$. If $\Delta C_1 \ge 0$, then directly applying $h_1$ on $F_1$ contradicts with incentive compatibility. Otherwise, for $0 < \mu_1 < \mu_2 \ll 1$ s.t. $\mu_1 \Delta C_1 + \mu_2 \Delta C_2 = 0$, the change in expected utility by reporting according to $h_1$ w.p. $\mu_1$ on $F_1$, according to $h_2$ w.p. $\mu_2$ on $F_2$, and otherwise truthfully is given by
    \begin{align*}
        & \mu_1 \Delta U_1 + \mu_2 \Delta U_2 \\
        > & - \mu_1 r_0 \cdot \Delta C_1 - \mu_2 r_0 \cdot \Delta C_2 \\
        = & 0.
    \end{align*}
    This contradicts with incentive compatibility.
\end{proof}

Given Lemma \ref{measurable} - \ref{ic-iff-new}, the proof of Theorem \ref{taxation} is quite straightforward:
\begin{proof}[Proof of Theorem \ref{taxation}]
    In fact, each of the two conditions in the theorem is equivalent to a corresponding condition in Lemma \ref{ic-iff-new}. We only give proof for equivalence of the second condition, because equivalence of the first condition is almost identical in shape and idea with the standard taxation principle in quasi-linear player case. 

    Suppose the second condition of the theorem holds. Then, for almost every $v \in V$, equation (\ref{new-sup-g}) is naturally satisfied for all $v' \in V$, and we have $r>0$ only if the ex-ante constraint is binding. By Lemma \ref{ic-iff-new}, $(x,p)$ is incentive compatible. This proves sufficiency of the condition.

    Suppose $(x,p)$ is incentive compatible by satisfying the second condition of Lemma \ref{ic-iff-new}. Then, for almost every $v$, for every $v'$ s.t. $x(v') = x(v)$, (\ref{new-sup-g}) gives 
    \[
        (c_1+rc_2) \cdot p(v') \ge (c_1+rc_2) \cdot p(v).
    \]
    Consider $p_x$ defined by:
    \[
        p_x(q) \coloneqq \inf_{v: x(v) = q} (c_1+rc_2)p(v)
    \]
    Then, for almost every $v \in V$, we have $(c_1+rc_2)p(v) \le p_x(x(v))$. By definition of infimum, for almost every $v \in V$, we have $p_x(x(v)) = (c_1+rc_2)p(v)$. With such a $p_x$, inequality (\ref{new-sup-g}) implies equation (\ref{taxed-maximization}) for almost every $v \in V$, which finishes our proof.
\end{proof}

\section{Differential properties of incentive compatible mechanisms}\label{diff-prop}
For the rest of this paper, we make the following assumption on $\rho$:
\begin{assume}\label{zero_boundar}
    The boundary of $V$ has zero probability, that is, $\rho(\partial V) = 0$.
\end{assume}
Therefore, when taking any expectation over $V$ with respect to $\rho$, it suffices to consider only $\mathring V$, the interior of $V$. Note that the assumption implies that $\mathring V \ne \emptyset$.

And we make the following assumptions, that the terms in the player's utility and constraint are all differentiable:
\begin{assume} \label{differentiable}
    For every $v \in \mathring V$, the single variable functions $u(x(\cdot), v), f(x(\cdot), v)$ and $p(\cdot)$ are differentiable at $v$.
\end{assume}
Therefore, we can define the vector-valued functions $\hat u$ and $\hat f$ by
\[
    \hat u(v) \coloneqq \nabla_{v'} [u(x(v'), v) - c_1 \cdot p(v')] |_{v'=v},\ \  \forall v \in \mathring V
\]
and 
\[
    \hat f(v) \coloneqq \nabla_{v'} [f(x(v'), v) - c_2 \cdot p(v')] |_{v'=v}, \ \ \forall v \in \mathring V.
\]

The following theorem gives necessary condition of incentive compatibility in terms of the mechanism's differential properties:
\begin{theorem} \label{ic-necessary}
    Under Assumption \ref{continuity}, \ref{zero_boundar} and \ref{differentiable}, the interim rule pair $(x,p)$ is incentive compatible only if at least one of the following two conditions hold:
    \begin{enumerate}
        \item The ex-ante constraint is binding when reporting truthfully, and for almost every $v \in \mathring V$, $\hat f(v) = \mathbf{0}$.
        \item There exists constant $r \in \mathbb{R}_{\ge 0}$ s.t. for almost every $v \in \mathring V$, $\hat u(v) + r \cdot \hat f(v) = 0$, and if $r>0$, the ex-ante constraint is binding when reporting truthfully.
    \end{enumerate}
\end{theorem}

To prove Theorem \ref{ic-necessary}, we first state the following useful lemma. 
\begin{lemma}\label{diff-exists}
    Given bivariate function $g: V \times V \rightarrow \mathbb{R}$ s.t. for any $v \in \mathring V$, the single variable function $g(\cdot, v)$ is differentiable at $v$, suppose there exists some $E \subseteq \mathring V$ with $\rho(E)>0$ s.t. for any $v\in E$, we have $\nabla_{v'} g(v', v) |_{v'=v} \ne \mathbf{0}$, then for any $v \in E$, there exists $v' \in V$ s.t $g(v', v) > g(v, v)$.
\end{lemma}
\begin{proof}
    For any $v \in E$, since $v \in \mathring V$, there exists $\epsilon_0 > 0$ s.t. the open ball $B(v, \epsilon_0) \subseteq V$. $\nabla_{v'} g(v', v) |_{v'=v} \ne \mathbf{0}$ implies the existence of some unit vector $y \in S^{d-1}$ s.t. $\lim_{\epsilon \downarrow 0} \frac{g(v+\epsilon y,v) - g(v,v)}{\epsilon}>0$. Therefore, we can find $\epsilon_1 < \epsilon_0$ s.t. $g(v+\epsilon_1 y,v) > g(v, v)$.
\end{proof}

Then we are ready to prove Theorem \ref{ic-necessary}. 
\begin{proof}[Proof of Theorem \ref{ic-necessary}]
    Suppose $(x,p)$ is incentive compatible, but neither of the two conditions hold. 

    If the constraint is unbinding when reporting truthfully, by Assumption \ref{zero_boundar}, the second condition is violated only if there exists some $E \subseteq \mathring V$ with $\rho(E)>0$ s.t. for every $v \in E$, $\hat u(v) \ne \mathbf{0}$. By Lemma \ref{diff-exists} for any $v \in E$ there exists $v'\in V$ s.t. $u(x(v'), v) - c_1 \cdot p(v') > u(x(v), v) - c_1 \cdot p(v)$. By Lemma \ref{ic-iff-new} such $(x,p)$ is not incentive compatible, contradiction.

    If the constraint is binding when reporting truthfully, by Assumption \ref{zero_boundar}, the first condition of Theorem \ref{ic-necessary} is violated only if there exists some $E \subseteq \mathring V$ with $\rho(E)>0$ s.t. for every $v \in E$, $\hat f(v) \ne \mathbf{0}$. By Lemma \ref{diff-exists} for any $v \in E$ there exists $v'\in V$ s.t. $f(x(v'), v) - c_2 \cdot p(v') > f(x(v), v) - c_2 \cdot p(v)$. Such $(x,p)$ violates the first condition of Lemma \ref{ic-iff-new}. The second condition of Theorem \ref{ic-necessary} is violated only if for any $r \ge 0$, there exists some $E \subseteq \mathring V$ with $\rho(E)>0$ s.t. for every $v \in E$, $\hat u(v) + r \cdot \hat f(v) \ne \mathbf{0}$. By Lemma \ref{diff-exists} for any $v \in E$ there exists $v'\in V$ s.t. 
    {\small
    \begin{equation}
        \begin{aligned}
            & [u(x(v'), v) - c_1 \cdot p(v')] - [u(x(v), v) - c_1 \cdot p(v)] \\
            > & -r \cdot [f(x(v'), v) - c_2 \cdot p(v')] + r \cdot [f(x(v), v) - c_2 \cdot p(v)],
        \end{aligned}
    \end{equation}
    }%
    which violates the second condition of Lemma \ref{ic-iff-new}. Therefore such $(x,p)$ is not incentive compatible, contradiction.

    We conclude that, to make $(x,p)$ incentive compatible, at least one of the two conditions in Theorem \ref{ic-necessary} holds.
\end{proof}
Theorem \ref{ic-necessary} gives necessary condition of incentive compatibility in terms of differential properties of the interim rules $(x, p)$ and functions $(u, f)$ given by the player model. Using Theorem \ref{ic-necessary} as a starting point, the next section gives another characterization of incentive compatibility under some additional assumptions, and the characterization is identical in shape with the payment identity in classic mechanism design \cite{10.2307/3689266,MANELLI2007153}.

\section{Quasi-linear player model}\label{char-2}
In this section, we restrict our attention to $u(\cdot, \cdot)$ and $f(\cdot, \cdot)$ which are linear in their second argument:
\begin{assume}\label{assum_linear}
    There exist functions $u^*, f^*: Q \rightarrow \mathbb{R}^d$ s.t. for any $v \in V$ and any $q \in Q$, we have
    \begin{align*}
        & u(q, v) = u^*(q) \cdot v,\\
        & f(q, v) = f^*(q) \cdot v.
    \end{align*}
    Moreover, the $u^*$ and $f^*$ are differentiable on $\mathring V$, and can be continuously extended to $V$.
\end{assume}

For technical simplicity, we also include the following assumptions. The first assumption holds for any $\rho$ which can be expressed by a positive probability density function on $V$. With these assumptions, the first condition of Theorem \ref{ic-necessary} is ruled out, and we can refer to $\hat u$ or $\hat f$ as if they are defined on $V$. 
\begin{assume}\label{dense}
    $V$ is connected. Moreover, for any $E \subseteq V$ with $\rho(E)=1$, $E$ is dense in $V$.
\end{assume}
\begin{assume}\label{nonzero-constraint-gradient}
    There exists some $E \subseteq V$ with $\rho(E)>0$ s.t. for any $v \in V$, $\hat f(v) \ne \mathbf{0}$,
\end{assume}
\begin{assume} \label{continous_diff}
    $\hat u$ and $\hat f$ are continuous on $\mathring V$, and can be continuously extended to $V$.
\end{assume}

We define the surrogate outcome function $\tilde u$ and surrogate utility function $\tilde U$, which are central in our last characterization of incentive compatibility.
\begin{define}\label{surrogate}
    The surrogate outcome $\tilde u : Q \rightarrow \mathbb{R}^d$ is defined to be:
    \[
        \tilde u(q) \coloneqq u^*(q) + r \cdot f^*(q).
    \]
    The surrogate utility $\tilde U: V \rightarrow \mathbb{R}$ is defined to be:
    \[
        \tilde U(v) \coloneqq \tilde u(x(v))\cdot v - (c_1+rc_2) \cdot p(v)
    \]
\end{define}
Note that the global constant $r$ is a parameter of surrogate outcome and surrogate utility, and such an $r$ exists if and only if $(x,p)$ is incentive compatible. When $r=0$, the surrogate utility is equal to the player's original utility. When $r>0$, the surrogate utility is a linear combination of a utility term and a constraint term, which again mirrors the results in auto-bidding literature.

With the additional assumptions, we are ready to give our second characterization of incentive compatibility.
\begin{theorem} \label{ic-iff}
    Under Assumption \ref{continuity} - \ref{continous_diff}, the interim rule pair $(x,p)$ is incentive compatible if and only only if, for surrogate outcome $\tilde u$ and surrogate utility $\tilde U$ given by Definition \ref{surrogate}, we have:
    \begin{itemize}
        \item $\tilde u$ and $\tilde U$ satisfy $\nabla \tilde U(v) = \tilde u(x(v)), \ \  \forall v \in V,$ AND
        \item $\tilde U$ is convex, AND
        \item if $r>0$, the ex-ante constraint is binding when reporting truthfully.
    \end{itemize}
\end{theorem}
\begin{proof}
    By Theorem \ref{ic-necessary} and Assumption \ref{nonzero-constraint-gradient}, $(x,p)$ is incentive compatible only if there exists constant $r \in \mathbb{R}_{\ge 0}$ s.t. for almost every $v \in \mathring V$, $\hat u(v) + r \cdot \hat f(v) = 0$, and if $r>0$, the ex-ante constraint is binding when reporting truthfully.

    By Assumption \ref{assum_linear}, $\hat u$ and $\hat f$ can be given by:
    \begin{align*}
    & \hat u(v) = \nabla[u^*(x(v))] \cdot v - c_1 \cdot \nabla p(v) \\
    & \hat f(v) = \nabla[f^*(x(v))] \cdot v - c_2 \cdot \nabla p(v) 
    \end{align*}
    For almost every $v \in \mathring V$, $\hat u(v) + r \cdot \hat f(v) = 0$. Then, equation
    {\small
    \begin{equation} \label{first-order}
        \nabla [u^*(x(v))] \cdot v - c_1 \nabla p(v) + r \cdot \nabla [f^*(x(v))] \cdot v - r c_2 \cdot \nabla p(v) = 0
    \end{equation}
    }%
    holds for almost every $v$. By Assumption \ref{zero_boundar}, \ref{dense} and \ref{continous_diff}, equation (\ref{first-order}) holds for almost every $v$ if and only if it holds for all $v \in V$. Note that by definition of surrogate outcome, (\ref{first-order}) can be further rewritten as
    \[
        \nabla [\tilde u(x(v))] \cdot v = (c_1 + rc_2) \nabla p(v)
    \]
    Consider an arbitrary point $v_0 \in V$. For any $v \in V$ and any path $L$ from $v_0$ to $v$, by integral by parts we have
    \[
        [\tilde u(x(t)) \cdot t]|_{v_0}^v - \int_L \tilde u(x(t)) \cdot \dif t = (c_1 + rc_2) \cdot p(t)|_{v_0}^{v}.
    \]
    By definition of surrogate utility, this can be rewritten as
    \[
        \tilde U(v) - \tilde U(v_0) = \int_L \tilde u(x(t)) \cdot \dif t,\ \ \forall v_0, v \in V.
    \]
    Therefore, $\tilde U$ is a potential function on $V$ whose gradient is given by $\tilde u \circ x$. In differential form, we have
    \begin{equation} \label{multidim-payment-identity}
        \nabla \tilde U(v) = \tilde u(x(v)), \ \  \forall v \in V.
    \end{equation}

    By Lemma \ref{ic-iff-new}, $(x, p)$ is incentive compatible if and only if equation (\ref{new-sup-g}) holds for almost every $v$ and every $v'$. Note that equation (\ref{new-sup-g}) can be rewritten as 
    {\small
    \[
        [\tilde u(x(v')) \cdot v - (c_1 + rc_2) p(v')]  - [\tilde u(x(v)) \cdot v - (c_1 + rc_2) p(v)] \le 0. 
    \]
    }%
    By definition of $\tilde U$ and equation (\ref{multidim-payment-identity}), we have:
    \begin{align*}
        & [\tilde u(x(v')) \cdot v - (c_1 + rc_2) p(v')] \\
        & - [\tilde u(x(v)) \cdot v - (c_1 + rc_2) p(v)] \\
        = & (\tilde u(x(v')) - \tilde u(x(v))) \cdot v - (\tilde u(x(v')) \cdot v' -\tilde U(v')) \\
        & + (\tilde u(x(v)) \cdot v -\tilde U(v)) \\
        = & \tilde U(v') - \tilde U(v) - \nabla \tilde U(v') \cdot (v'-v) 
    \end{align*}
    Therefore, $(x, p)$ is truthful if and only if for almost every $v$, for every $v'$, we have $\tilde U(v') - \tilde U(v) - \nabla \tilde U(v') \cdot (v'-v) \le 0$. By assumption of continuous $\tilde u$, $x$ and $p$, and that such $v$ is dense in $V$, we conclude that $\tilde U(v') - \tilde U(v) - \nabla \tilde U(v') \cdot (v'-v) \le 0$ holds for all $v \in V$, which is equivalent to $\tilde U$ being convex.
\end{proof}

\section*{Acknowledgments}
We sincerely acknowledge Xiao Han, who is a researcher in Theoretical Probability, for his constructive advice on the proofs.

\bibliography{autobidding}


\end{document}